\documentclass[12pt]{article}
\usepackage{amsmath, amssymb,latexsym}
\usepackage{color}
\title{Quadratic $\P\T$--symmetric operators with real spectrum and similarity to self-adjoint operators}
\author{Emanuela Caliceti\\Dipartimento di Matematica \\ Universit\`a di Bologna\\ 40127 Bologna, Italy\\ \small Emanuela.Caliceti@unibo.it \and
Sandro Graffi \\Dipartimento di Matematica \\ Universit\`a di Bologna\\ 40127 Bologna, Italy\\ \small Sandro.Graffi@unibo.it \and
Michael Hitrik\\Department of Mathematics \\University of California \\ Los Angeles
\\CA 90095-1555, USA\\\small hitrik@math.ucla.edu \and
Johannes Sj\"ostrand\\IMB, Universit\'e de Bourgogne\\9, Av. A. Savary, BP 47870\\FR--21078 Dijon, France \\and UMR 5584 CNRS
\\\small johannes.sjostrand@u-bourgogne.fr}
\date{}

\def\wrtext#1{\relax\ifmmode{\leavevmode\hbox{#1}}\else{#1}\fi}
\def\abs#1{\left|#1\right|}
\def\begeq{\begin{equation}}
\def\endeq{\end{equation}}

\def\iint{\int\hskip -2mm\int}

\def\P{{\mathcal P}}
\def\T{{\mathcal T}}
\def\CC{{\mathcal C}}

\def\Re{{\rm Re\,}}
\def\Im{{\rm Im\,}}

\def\part#1{\frac{\partial}{\partial #1}}
\def\half{\frac{1}{2}}

\newcommand{\real}{\mbox{\bf R}}
\newcommand{\comp}{\mbox{\bf C}}

\newcommand{\nat}{\mbox{\bf N}}

\renewcommand{\Re}{\mbox{\rm Re\,}}
\renewcommand{\Im}{\mbox{\rm Im\,}}

\renewcommand{\exp}{\mbox{\rm exp\,}}

\newtheorem{dref}{Definition}[section]

\newtheorem{theo}[dref]{Theorem}
\newtheorem{prop}[dref]{Proposition}

\newenvironment{proof}{\vspace{.3cm}\noindent{{\em Proof:}}}{\hfill$\Box$}

\begin{document}
\maketitle

\vspace*{1cm}
\noindent
{\bf Abstract}: It is established that a $\P\T$-symmetric elliptic quadratic differential operator with real spectrum is
similar to a self-adjoint operator precisely when the associated fundamental matrix has no Jordan blocks.



\tableofcontents
\section{Introduction and statement of result}
\setcounter{equation}{0}
$\P\T$--symmetric operators in quantum mechanics (for general references on $\P\T$--symmetry and quantum mechanics, see e.g.\,~\cite{BBM}) are those operators
on $\real^n$ left invariant by a successive application of the parity operator $\P$, acting on  wave functions as
$$
(\P\psi)(x)=\psi((-1)^{j_1}x_1,\ldots,\, (-1)^{j_n}x_n),\quad j_k = 0,1,\,\,\wrtext{not all equal to}\,\, 0,
$$
and of the time-reversal symmetry, acting as $(\T\psi)(x)=\overline{\psi(x)}$.
The $\P\T$--symmetry of an operator is called {\it exact} if
its spectrum is purely real, see e.g.\,~\cite{BB}. Generally speaking, the reality of the spectrum, and thus the exact $\P\T$--symmetry, cannot follow by
a unitary equivalence to a
self-adjoint operator, since a $\P\T$--symmetric operator is not necessarily self-adjoint, and actually, the standard examples are not self-adjoint and
not even normal operators. However, it could follow by a {\it similarity} to a self-adjoint operator. Let us therefore
proceed now to discuss the notion of similarity for two unbounded linear operators. For the limited purposes of this paper, here we shall only consider a
rather particular abstract situation, introducing assumptions which will be satisfied in the specific instances below.

\medskip
\noindent
Let ${\cal H}_j$, $j=1,2$, be complex separable Hilbert spaces, and let ${\cal A}_j:{\cal H}_j \rightarrow {\cal H}_j$ be closed densely defined operators
such that for $j=1,2$, the spectrum ${\rm Spec}({\cal A}_j)\subset \comp$ is discrete, consisting of eigenvalues of finite algebraic multiplicity. When
$\lambda \in {\rm Spec}({\cal A}_j)$, we let $E_{\lambda}^{(j)}\subset {\cal D}({\cal A}_j)$ be the finite-dimensional space of generalized
eigenvectors, corresponding to ${\cal A}_j$, $\lambda$, so that
$$
E_{\lambda}^{(j)} = {\rm Ker}\left({\cal A}_j - \lambda\right)^N,
$$
if $N \geq N(\lambda,j)$ is large enough.

\medskip
\noindent
Let ${\cal S}_j \subset {\cal D}({\cal A}_j)$, $j=1,2$, be linear subspaces such that ${\cal A}_j ({\cal S}_j) \subset {\cal S}_j$, and
$E_{\lambda}^{(j)} \subset {\cal S}_j$, for each  $\lambda\in {\rm Spec}({\cal A}_j)$. Assume that $S: {\cal S}_1 \rightarrow {\cal S}_2$ is
a linear bijection such that
\begeq
\label{sim}
S {\cal A}_1  = {\cal A}_2 S\quad \wrtext{on}\,\, {\cal S}_1,
\endeq
or equivalently, such that
\begeq
\label{sim1}
{\cal A}_1 S^{-1} = S^{-1} {\cal A}_2\quad \wrtext{on}\,\, {\cal S}_2.
\endeq
We shall then say that the operators ${\cal A}_1$ and ${\cal A}_2$ are related by the similarity transformation $S$, or are {\it similar}.

\begin{prop}
The similar operators ${\cal A}_1$ and ${\cal A}_2$ are isospectral.
\end{prop}
\begin{proof}
Let $u\in E_{\lambda}^{(1)}$, so that $({\cal A}_1-\lambda)^N u=0$, for some $N\in \nat$. It follows that $u, ({\cal A}_1-\lambda)u,\ldots\, \in {\cal S}_1$,
and we can apply (\ref{sim}) to see that $({\cal A}_2-\lambda)^N Su=0$, so $Su \in E_{\lambda}^{(2)}$. Thus we have shown that
$S(E_{\lambda}^{(1)}) \subset E_{\lambda}^{(2)}$. The same argument applied to $S^{-1}$ gives $E_{\lambda}^{(1)} \supset S^{-1}(E_{\lambda}^{(2)})$, and we get
${\rm Spec}({\cal A}_1)$ = ${\rm Spec}({\cal A}_2)$, and $S(E_{\lambda}^{(1)}) = E_{\lambda}^{(2)}$, for all
$\lambda \in {\rm Spec}({\cal A}_1) = {\rm Spec}({\cal A}_2)$.
\end{proof}

\bigskip
\noindent
Notice that in the discussion above, the operator ${\cal A}_2$ may be self-adjoint on ${\cal H}_2$, even though the operator ${\cal A}_1$ need not be
self-adjoint on ${\cal H}_1$. It is therefore natural to ask whether a $\P\T$--symmetric operator with real spectrum is always similar to a self-adjoint one.
This problem has attracted considerable attention, notably by~\cite{Mo1},~\cite{Mo2},~\cite{Mo3}, in an abstract framework. Here instead, we shall consider the class of $\P\T$--symmetric elliptic quadratic differential operators acting on $L^2(\real^n)$, with real spectrum,
and establish necessary and sufficient conditions for their similarity to a self-adjoint operator, as above. A noticeable peculiarity of these conditions
is their classical nature, i.e. their dependence only on the classical flow generated by the classical quadratic hamiltonian, whose quantization yields the given
$\P\T$--symmetric quadratic differential operator. Let us also remark here that in general, the operator $S$ in (\ref{sim}) realizing the similarity between the
non-self-adjoint operator ${\cal A}_1$ and the self-adjoint operator ${\cal A}_2$, cannot be expected to be bounded with a bounded inverse, due
to the issue of pseudospectra for non-self-adjoint operators, see~\cite{DeSjZw},~\cite{Da1},~\cite{BaZn}.

\bigskip
\noindent
Let $q(x,\xi)$ be a complex-valued quadratic form on $\real^{2n} = \real_x^n \times \real^n_{\xi}$. Here $(x,\xi)\in \real^{2n}$ are the canonical coordinates,
so that $\{\xi_i,x_j\}=\delta_{ij}$, $1\leq i,j\leq n$. Throughout this work, we shall assume that the quadratic form $q$ is elliptic on $\real^{2n}$, in the sense
that $q(X)=0$, $X\in \real^{2n}$, if and only if $X=0$. An application of Lemma 3.1 of~\cite{Sj74} then shows, if $n>1$, that there exists
$z\in \comp\backslash\{0\}$ such that $\Re(zq)$ is positive definite. In the case when $n=1$, the same conclusion holds, provided that the range of
$q$ on $\real^2$ is not all of $\comp$, which will be assumed in the sequel. After a multiplication by $z$, we shall assume in what follows, as we may, that
$z=1$, so that $\Re q$ is positive definite. It follows that the range $\Sigma(q):= q(\real^{2n})$ of $q$ on $\real^{2n}$ is a closed angular sector with a
vertex at zero, contained in the union of $\{0\}$ and the open right half-plane.

\medskip
\noindent
We shall now introduce the assumption of the $\P\T$-symmetry of the quadratic symbol $q$. To that end, let
$\kappa: {\real}^n \rightarrow \real^n$ be linear and such that
\begeq
\label{eq_inv}
\kappa^2 = 1.
\endeq
Associated to the involution $\kappa$, we have the parity operator ${\cal P}$, given by
\begeq
\label{eq_par}
({\cal P} u)(x) = u(\kappa(x)),
\endeq
and the lift of $\kappa$ to the cotangent space, given by the linear involution,
\begeq
\label{eq_lift}
{\cal K}: \real^{2n}\rightarrow \real^{2n}, \quad {\cal K}(x,\xi)= (\kappa(x),-\kappa^{t}(\xi)).
\endeq

\medskip
\noindent
The $\P\T$-symmetry assumption on the quadratic symbol $q$ is of the following form,
\begeq
\label{PT}
\overline{q\circ{\cal K}} =q \Longleftrightarrow  \overline{q(\kappa(x),-\kappa^t(\xi))} =q(x,\xi),\quad (x,\xi)\in \real^{2n}.
\endeq
It follows, in particular, that the sector $\Sigma(q)$ is symmetric with respect to reflection in the real axis.

\bigskip
\noindent
Associated to the symbol $q$ is the corresponding operator ${\rm Op}^w(q)$ acting on $L^2(\real^n)$, formally defined as the Weyl quantization of $q(x,\xi)$,
\begeq
\label{WQ}
{\rm Op}^w(q)u(x) =  (2\pi)^{-n}\int\!\!\!\int_{{\bf R}^n\times{\bf R}^n} e^{i\xi\cdot(x-y)} q\left(\frac{x+y}{2},\xi\right)u(y)\,dy\,d\xi.
\endeq
Writing
\begeq
\label{eq1.1}
q(x,\xi) = \sum_{|\alpha+\beta|=2}q_{\alpha,\beta}x^\alpha \xi^\beta,
\endeq
we see that the operator ${\rm Op}^w(q)$ is a differential operator of the form
$$
{\rm Op}^w(q) = \sum_{|\alpha+\beta|=2} q_{\alpha,\beta} \frac{x^{\alpha} D_x^{\beta} + D_x^{\beta}\,x^{\alpha}}{2}.
$$

\medskip
\noindent
Let us recall from~\cite{Ho95} that the maximal closed realization of the operator ${\rm Op}^w(q)$, i.e.\, the operator on $L^2(\real^n)$ equipped with the domain
$$
\{u\in L^2(\real^n);\, {\rm Op}^w(q)u\in L^2(\real^n)\} = \{u \in L^2(\real^n);\,x^{\alpha} D_x^{\beta} u \in L^2(\real^n),\, \abs{\alpha+\beta}\leq 2\},
$$
agrees with the graph closure of its restriction to ${\cal S}(\real^n)$,
$$
{\rm Op}^w(q): {\cal S}(\real^n) \rightarrow {\cal S}(\real^n).
$$
Furthermore, the spectrum of ${\rm Op}^w(q)$ is discrete, and its precise description will be recalled below.

\medskip
\noindent
The assumption (\ref{PT}) implies the $\P\T$--symmetry of the operator ${\rm Op}^w(q)$, i.e. the commutation property $[{\rm Op}^w(q),\P\T]=0$. Here the
parity operator ${\cal P}$ has been introduced in (\ref{eq_par}). Indeed, we have
\begin{eqnarray*}
&&
\P\circ {\rm Op}^w(q) \circ \P^{-1} = {\rm  Op}^w(q(\kappa(x),\kappa^t(\xi))); \quad \T\circ {\rm Op}^w (q) \circ \T^{-1}={\rm Op}^w(\overline{q}(x,-\xi)),
\\
&&
 \P\T\circ {\rm Op}^w(q) \circ (\P\T)^{-1}={\rm Op}^w(\overline{q}(\kappa(x),-\kappa^t(\xi))).
\end{eqnarray*}
It follows that the spectrum of ${\rm Op}^w(q)$ is symmetric with respect to the real axis.

\bigskip
\noindent
The essential role in what follows will be played by the fundamental matrix $F$ of the quadratic form $q$. When recalling the definition of $F$
following Section 21.5 of~\cite{H_book}, we let
$$
\sigma((x,\xi),(y,\eta)) = \xi \cdot y - \eta \cdot x, \quad (x,\xi) \in \real^{2n},\quad (y,\eta)\in \real^{2n}, $$
be the canonical symplectic form on $\real^{2n}$, which extends to the complex symplectic form on $\comp^{2n}$. Letting also $q(X,Y)$ stand for the polarization
of $q$, viewed as a symmetric bilinear form on $\comp^{2n}$, we define the $2n\times 2n$ fundamental matrix $F$ by the identity
\begeq
\label{eq1.3}
q(X,Y) = \sigma(X,FY),\quad X,Y\in \comp^{2n}.
\endeq
We notice that the fundamental matrix $F$ is skew-symmetric with respect to $\sigma$, and furthermore, following~\cite{H_book},
we see that in the canonical coordinates $(x,\xi)$, it is given by
\begeq
\label{eq1.4}
F=\frac{1}{2} \left(\begin{array}{ll} q''_{\xi,x} & q''_{\xi,\xi}\\ -q''_{x,x} & -q''_{x,\xi} \end{array}\right).
\endeq

\bigskip
\noindent
We can now state the main result of this work.
\begin{theo}
\label{mth}
Let $q: \real_x^{n}\times \real^n_{\xi} \rightarrow \comp$ be a quadratic form, such that $\Re q$ is positive definite. Let $\kappa: \real^n \rightarrow
\real^n$ be linear and such that $\kappa^2=1$. Assume the property of the ${\cal PT}$--symmetry,
$$
\overline{q(\kappa(x),-\kappa^t(\xi))} = q(x,\xi),\quad (x,\xi)\in \real^{2n},
$$
and that the spectrum of the quadratic operator ${\rm Op}^w(q)$ is purely real. Then the operator ${\rm Op}^w(q)$ is similar to a self-adjoint operator,
in the sense of the discussion in the beginning of this section,
precisely when the fundamental matrix $F$ of $q$ has no Jordan blocks.
\end{theo}

\bigskip
\noindent
{\bf Remarks}.
\begin{enumerate}
\item
Let us consider the Hamilton vector field of $q$,
\begeq
H_q =q'_{\xi}\cdot \partial_x-q'_x\cdot \partial_\xi.
\endeq
According to (\ref{eq1.4}), we have
$$
F Y = \frac{1}{2} H_q(Y).
$$
Therefore the possibility of establishing a similarity between a ${\cal PT}$--symmetric elliptic quadratic operator with purely real spectrum and a
self-adjoint operator involves only the underlying classical flow, i.e. is determined by a purely classical condition.
\item Examples of $\P\T$--symmetric non-diagonalizable hamiltonians with real spectrum can be found in \cite{CGrSj}, as well as in
\cite{ACS}, \cite{SAC}, \cite{CIN}. See also the discussion in Section 4.
\item In Theorem \ref{mth} we have only considered the case when the quadratic form $q$ is elliptic on $\real^{2n}$. It seems quite likely, however, that the
result of Theorem \ref{mth} can be extended to a suitable class of partially elliptic quadratic operators, namely the one studied in~\cite{HiPr1}, \cite{Viola2}.
We leave this extension open until needed.
\end{enumerate}

\medskip
\noindent
The plan of this note is as follows. In Section 2, we establish a symmetry pro\-perty of the fundamental matrix $F$, as a consequence of the assumption of
the ${\cal PT}$--symmetry of $q$. The proof of Theorem \ref{mth} is then carried out in Section 3, using the techniques of FBI--Bargmann
transformations with quadratic phases~\cite{Sj82}, which, in the quadratic case, allow us to construct the similarity transformation $S$ explicitly.
Section 4 is devoted to the discussion of an example, due to~\cite{CIN}, of an elliptic quadratic $\P\T$--symmetric operator with real spectrum,
for which the fundamental matrix possesses Jordan blocks. The appendix A gives a concise exposition of aspects of the theory of positive
Lagrangian planes, required in the proof of the main result, while in the appendix B, the general theory of quadratic Fourier integral operators in the
complex domain, arising when quantizing complex linear canonical transformations, is developed.

\bigskip
\noindent
{\bf Acknowledgements}. Johannes Sj\"ostrand acknowledges support from the French ANR grant NOSEVOL: ANR 2011 BS01019 01. 
We would also like to thank the referees for very helpful remarks and suggestions. This led to a generalization of our main result.

\section{$\P\T$--invariance condition for the fundamental matrix}
\setcounter{equation}{0}
In this section, we let $q: \real^{2n} \rightarrow \comp$ be a quadratic form such that $\Re q$ is positive definite. It follows therefore from
(\ref{eq1.3}) that the eigenvalues of the corresponding fundamental matrix $F$ avoid the real axis, and in general we know from Section 21.5 of~\cite{H_book}
that if $\lambda$ is an eigenvalue of $F$, then so is $-\lambda$, and the algebraic multiplicities agree. Let us also recall from~\cite{H_book}
that the eigenvalues of $F$ belong to the set $i\Sigma(q) \cup -i\Sigma(q)$, and let $\lambda_1,\ldots, \lambda_n$ be the eigenvalues of $F$,
counted according to their multiplicity, such that $\lambda_j/i\in \Sigma(q)$, $j=1,\ldots, n$. From~\cite{Boutet},~\cite{Sj74}, we know that
the spectrum of the operator ${\rm Op}^w(q)$ is given by the eigenvalues of the form
\begeq
\label{eq2.1}
\sum_{j=1}^n \frac{\lambda_j}{i} \left(2\nu_{j,\ell}+1\right),\quad \nu_{j,\ell}\in \nat\cup\{0\}.
\endeq
Notice that ${\rm Spec}({\rm Op}^w(q))\subset \Sigma(q)$.

\bigskip
\noindent
Let us recall the real linear involution ${\cal K}$, introduced in (\ref{eq_lift}). Writing ${\cal K} = J \circ {\cal K}_1$, where $J(x,\xi) = (x,-\xi)$ and
${\cal K}_1(y,\eta) = (\kappa(y), (\kappa^{t})^{-1}(\eta))$ is symplectic, we see, using also that the involution $J$ is skew-symmetric with respect to the
symplectic form $\sigma$, that the map ${\cal K}$ is antisymplectic in the sense that
$$
\sigma({\cal K}X, {\cal K}Y) = - \sigma(X,Y),\quad X,Y\in \real^{2n}.
$$
In particular,
$$
\sigma({\cal K}X,Y) = -\sigma({\cal K}^2 X, {\cal K}Y) = -\sigma(X,{\cal K}Y).
$$
We shall also let $\CC:\comp^{2n} \rightarrow \comp^{2n}$ stand for the involution defined by the operation of complex conjugation. We have the following result.

\begin{prop}
Assume that the quadratic form $q$ with $\Re q>0$ is such that $\overline{q({\cal K}(x,\xi))} = q(x,\xi)$, $(x,\xi)\in \real^{2n}$. Then the fundamental matrix
$F$ of $q$ satisfies
\begeq
\label{PTF}
F = - {\cal K} \CC \circ F\circ {\cal K} \CC.
\endeq
\item
Furthermore, the eigenvalues of $F$ are symmetric with respect to the imaginary axis.
\end{prop}

\begin{proof}
Let us write,
$$
q(x,\xi) = Ax \cdot x + 2  Bx \cdot \xi + C\xi \cdot \xi.
$$
Here $A$, $B$, and $C$ are complex $n\times n$ matrices, with $A$ and $C$ symmetric. The condition of the $\P\T$--symmetry,
$$
\overline{q(\kappa(x),-\kappa^t(\xi))} = q(x,\xi),\quad (x,\xi)\in \real^{2n},
$$
implies then that the matrices $A$, $B$, and $C$ satisfy
\begeq
\label{eq_ABC}
\overline{A} = \kappa^t A \kappa, \quad \overline{B} = -\kappa B \kappa, \quad \overline{C} = \kappa C \kappa^t.
\endeq

\medskip
\noindent
When establishing the property (\ref{PTF}), we introduce the symmetric bilinear form associated to $q$, given by
$$
q((x,\xi),(y,\eta)) = Ax \cdot y + Bx \cdot \eta + By \cdot \xi + C\xi \cdot \eta,\quad (x,\xi)\in \comp^{2n},\,\,
(y,\eta)\in\comp^{2n}.
$$
Using (\ref{eq_ABC}), we see that
$$
\overline{q({\cal K}(x,\xi),{\cal K}(y,\eta))} = q(\CC(x,\xi),\CC(y,\eta)),
$$
and therefore, from (\ref{eq1.3}) we get
$$
\sigma(\CC X, F\CC Y) = \overline{\sigma({\cal K}X, F {\cal K} Y)} = - \overline{\sigma(X, {\cal K} F {\cal K} Y)} = -
\sigma(\CC X, \CC {\cal K}F {\cal K}Y).
$$
Here we have also used that the map ${\cal K}$ is skew-symmetric with respect to $\sigma$. The identity (\ref{PTF}) follows, and it only remains to check the
symmetry of the eigenvalues of $F$ with respect to the imaginary axis. To that end, assume that $\lambda \in \comp$ is such that $F X = \lambda X$,
$X \in \comp^{2n}$. If $Y = {\cal K} {\cal C} X$, then $\left(F\circ {\cal K} {\cal C}\right) Y = \lambda {\cal K} {\cal C} Y$. Applying the antilinear map
${\cal K} {\cal C}$ and using (\ref{PTF}), we conclude that $F Y = -\overline{\lambda} Y$. The proof is complete.
\end{proof}

\bigskip
\noindent
{\bf Remark 2.2}. Combining Proposition 2.1 with the explicit description of the spectrum of the ${\cal PT}$--symmetric operator ${\rm Op}^w(q)$, given by
(\ref{eq2.1}), we conclude that if the spectrum is purely real, then the eigenvalues $\lambda_j$, $1 \leq j \leq n$, are purely imaginary.

\section{Similarity transformation and proof of Theorem \ref{mth}}
\setcounter{equation}{0}
Assume that $q: \real^{2n}\rightarrow \comp$ is a quadratic form such that $\Re q>0$, and let $F$ be the fundamental matrix of $q$, introduced in (\ref{eq1.3}).
When $\lambda \in {\rm Spec}(F)$, we let
\begeq
\label{eq3.1}
E_{\lambda} = {\rm Ker}((F-\lambda)^{2n})\subset \comp^{2n}
\endeq
be the spectral subspace corresponding to the eigenvalue $\lambda$. According to Lemma 21.5.2 of~\cite{H_book}, the complex symplectic form $\sigma$ is
nondegenerate viewed as a bilinear form on $E_{\lambda}\times E_{-\lambda}$.

\medskip
\noindent
Let us introduce the unstable linear manifold for the Hamilton flow of the quadratic form $i^{-1}q$, given by
\begeq
\label{eq3.2}
\Lambda^+  := \bigoplus_{{\rm Im}\, \lambda>0} E_{\lambda}\subset \comp^{2n}.
\endeq
According to Proposition 3.3 of~\cite{Sj74}, the complex Lagrangian plane $\Lambda^+$ is strictly positive in the sense that
\begeq
\label{eq3.3}
\frac{1}{i} \sigma(X,\overline{X})>0,\quad 0\neq X\in \Lambda^+.
\endeq
We refer to Appendix A for a discussion of positivity conditions in $\comp^{2n}$. Defining also
\begeq
\label{eq3.4}
\Lambda^-  = \bigoplus_{{\rm Im}\, \lambda<0} E_{\lambda}\subset \comp^{2n},
\endeq
we know from the arguments of~\cite{Sj74} that the complex Lagrangian plane $\Lambda^-$ is strictly negative in the sense that
\begeq
\label{eq3.5}
\frac{1}{i} \sigma(X,\overline{X})<0,\quad 0\neq X\in \Lambda^-.
\endeq
It also follows from (\ref{eq1.3}) that the quadratic form $q$ vanishes when restricted to $\Lambda^{\pm}$.

\bigskip
\noindent
Our proof of Theorem \ref{mth} will proceed by exhibiting a complex linear canonical transformation which will reduce $\Lambda^+$ to
$\{(x,\xi)\in \comp^{2n};\, \xi=0\}$ and $\Lambda^-$ to $\{(x,\xi)\in \comp^{2n};\, x=0\}$. We shall then be able to quantize the canonical
transformation by means of an FBI--Bargmann transform~\cite{Sj82}, which will essentially provide us with the sought similarity operator.
Let us notice here that the use of such canonical transformations has a long and rich tradition, see for
instance~\cite{MeSj78},~\cite{HeSj84},~\cite{Sj86}. The discussion below will follow Section 2 of~\cite{HiSjVi} closely.

\bigskip
\noindent
Although not necessary, it will be convenient to simplify $q$ first by means of a suitable real linear canonical transformation. When doing so, we observe
that an application of Example A.6 in Appendix A shows that the negative Lagrangian $\Lambda^-$ is of the form
$$
\eta = A_-y,\quad y\in \comp^n,
$$
where the complex symmetric $n\times n$ matrix $A_-$ is such that $\Im A_-<0$. Here $(y,\eta)$ are the standard canonical coordinates on
$T^*\real^n_y=\real^n_y\times \real^n_{\eta}$, extended to the complexification $T^*\comp^n_y$. Using the real linear canonical transformation
$(y,\eta)\mapsto (y, \eta - (\Re A_-)y)$, we reduce $\Lambda^-$ to the form $\eta = i\Im A_-y$, and by a diagonalization of $\Im A_-$, we
obtain the standard form $\eta = -iy$. After this real linear canonical transformation and the conjugation of the quadratic operator ${\rm Op}^w(q)$
by means of the corresponding metaplectic operator, unitary on $L^2(\real^n)$, we may assume that $\Lambda^-$ is of the form
\begeq
\label{eq3.6}
\eta = -iy,\quad y\in \comp^n.
\endeq
We refer to the Appendix to Chapter 7 of~\cite{DiSj} for the introduction to and basic properties of metaplectic operators. Notice also that according
to Example A.6, in the new real symplectic coordinates, extended to the complexification, the positive complex Lagrangian $\Lambda^+$ is of the form
\begeq
\label{eq3.7}
\eta = A_+ y,\quad \Im A_+>0.
\endeq

\bigskip
\noindent
Let
\begeq
\label{eq3.8}
B = B_+ = (1-iA_+)^{-1}A_+,
\endeq
and notice that the matrix $B$ is symmetric. The holomorphic quadratic form
\begeq
\label{eq3.9}
\varphi(x,y)  = \frac{i}{2}(x-y)^2 - \frac{1}{2} Bx \cdot x,\quad (x,y)\in \comp^{2n},
\endeq
satisfies $\Im \varphi''_{y,y}>0$ and ${\rm det}\,\varphi''_{x,y}\neq 0$. It gives rise therefore to the FBI--Bargmann transformation,
\begeq
\label{eq3.10}
Tu(x) = C \int_{{\bf R}^n} e^{i\varphi(x,y)} u(y)\,dy,\quad x\in \comp^n,\quad C>0,
\endeq
and following~\cite{Sj82},~\cite{Sj95}, the operator $T$ is to be viewed as a Fourier integral operator, with the associated complex linear canonical
transformation of the form
\begeq
\label{eq3.11}
\kappa_T: \comp^{2n} \ni (y,-\varphi'_y(x,y))\mapsto (x,\varphi'_x(x,y))\in \comp^{2n}.
\endeq
Using (\ref{eq3.8}), (\ref{eq3.9}), we see that $\kappa_T$ is given by
\begeq
\label{eq3.12}
\kappa_T: (y,\eta)\mapsto (x,\xi)=(y-i\eta, \eta+iB\eta-By),
\endeq
and therefore, we have $\kappa_T\left(\Lambda_-\right) = \{(x,\xi)\in \comp^{2n};\, x=0\}$, while $\kappa_T(\Lambda^+)$ is given by the equation
$\{(x,\xi)\in \comp^{2n}; \xi=0\}$.

\medskip
\noindent
We know from~\cite{Sj95} that for a suitable choice of $C>0$ in (\ref{eq3.10}), the map $T$ is unitary,
\begeq
\label{eq3.13}
T: L^2(\real^n)\rightarrow H_{\Phi_0}(\comp^n),
\endeq
where
$$
H_{\Phi_0}(\comp^n) = {\rm Hol}\,(\comp^n)\cap L^2(\comp^n; e^{-2\Phi_0}L(dx)),
$$
and $\Phi_0$ is a strictly plurisubharmonic quadratic form on $\comp^n$, given by
\begeq
\label{eq3.14}
\Phi_0(x) = {\rm sup}\,_{y\in {\bf R}^n} \left(-\Im \varphi(x,y)\right) = \frac{1}{2} \left(\left(\Im x\right)^2 + \Im \left(Bx \cdot x \right)\right).
\endeq
From~\cite{Sj95}, we recall also that the canonical transformation $\kappa_T$ in (\ref{eq3.11}) maps $\real^{2n}$ bijectively onto
\begeq
\label{eq3.15}
\Lambda_{\Phi_0}=\left\{\left(x,\frac{2}{i}\frac{\partial \Phi_0}{\partial x}(x)\right); x\in \comp^n\right\}.
\endeq
The positivity of $\{(x,\xi)\in \comp^{2n};\, \xi=0\}$ with respect to $\Lambda_{\Phi_0}$ implies, in view of Proposition A.9 in Appendix A,
that the quadratic form $\Phi_0$ is strictly convex, so that
\begeq
\label{eq3.16}
\Phi_0(x) \sim \abs{x}^2,\quad x\in \comp^n.
\endeq

\bigskip
\noindent
{\bf Example: The standard Bargmann transformation.} Following~\cite{Bar61}, let us consider a complex integral transform of the form (\ref{eq3.10}),
where the phase function $\varphi(x,y)$ is given by
$$
\varphi(x,y) = i \left(\frac{x^2}{2} + \sqrt{2}xy + \frac{y^2}{2}\right).
$$
The corresponding canonical transformation, defined as in (\ref{eq3.11}), is then of the form
$$
(y,\eta) \mapsto \frac{1}{\sqrt{2}} (i\eta -y, -\eta + iy),
$$
while the associated quadratic weight function is given by $\Phi(x) = (1/2) \abs{x}^2$.

\bigskip
\noindent
Returning to the FBI--Bargmann transformation given by (\ref{eq3.9}), (\ref{eq3.10}), let us next recall the exact Egorov theorem,~\cite{Sj95},
\begeq
\label{eq3.17}
T {\rm Op}^w(q) u = {\rm Op}^w(\widetilde{q}) Tu,\quad u\in {\cal S}(\real^n),
\endeq
where $\widetilde{q}$ is a quadratic form on $\comp^{2n}$ given by $\widetilde{q} = q\circ \kappa_T^{-1}$. It follows therefore that
\begeq
\label{eq3.18}
\widetilde{q}(x,\xi) = Mx \cdot \xi,
\endeq
where $M$ is a complex $n\times n$ matrix. We have
$$
H_{\widetilde{q}} =  Mx \cdot \partial_x - M^{t}\xi \cdot \partial_{\xi},
$$
and using (\ref{eq1.4}), we see that the corresponding Hamilton map
$$
\widetilde{F} =\half \begin{pmatrix}
M & 0 \\
0 & -M^{t}
\end{pmatrix}
$$
maps $(x,0)\in \kappa_T(\Lambda^+)$ to $(1/2)(Mx,0)$. Now the maps $F$ and $\widetilde{F}$ are isospectral, and we conclude that, with the agreement of
algebraic multiplicities, the following holds,
\begeq
\label{eq3.19}
{\rm Spec}(M) = {\rm Spec}(2F)\cap \{\Im \lambda>0\}.
\endeq

\bigskip
\noindent
The final step in the normal form construction is the reduction of the matrix $M$ in (\ref{eq3.18}) to its Jordan normal form. Such a reduction
is implemented by considering a complex linear canonical transformation of the form
\begeq
\label{eq3.20}
\kappa_{C}: \comp^{2n}\ni (x,\xi)\mapsto (C^{-1}x, C^{t}\xi)\in \comp^{2n},
\endeq
where $C$ is a suitable invertible complex $n\times n$ matrix. On the operator level, associated to the transformation in (\ref{eq3.20}), we have the
operator
\begeq
\label{eq3.21}
U_C: u(x)\mapsto \abs{{\rm det}\, C} u(Cx),
\endeq
which maps the space $H_{\Phi_0}(\comp^n)$ unitarily onto the space $H_{\Phi_1}(\comp^n)$, where $\Phi_1(x)=\Phi_0(Cx)$ is a strictly plurisubharmonic
quadratic form such that $\kappa_C(\Lambda_{\Phi_0}) = \Lambda_{\Phi_1}$. We notice that the strict convexity property
\begeq
\label{eq3.22}
\Phi_1(x) \sim \abs{x}^2,\quad x\in \comp^n,
\endeq
remains valid.

\medskip
\noindent
We obtain therefore the following result, which only is a slight reformulation of Proposition 2.1 of~\cite{HiSjVi} and is closely related to the discussion
in Section 3 of~\cite{Sj74}.

\begin{prop}
Let $q: \real^{2n} \rightarrow \comp$ be a quadratic form, such that $\Re q>0$. The quadratic differential operator
$$
{\rm Op}^w(q): L^2(\real^n)\rightarrow L^2(\real^n),
$$
equipped with the domain
$$
{\cal D}({\rm Op}^w(q)) = \{u\in L^2(\real^n); \left(x^2+(D_x)^2\right) u\in L^2(\real^n)\},
$$
is unitarily equivalent to the quadratic operator,
$$
{\rm Op}^w(\widetilde{q}): H_{\Phi_1}(\comp^n)\rightarrow H_{\Phi_1}(\comp^n),
$$
with the domain
$$
{\cal D}({\rm Op}^w(\widetilde{q})) = \{u\in H_{\Phi_1}(\comp^n); (1+\abs{x}^2)u\in L^2_{\Phi_1}(\comp^n)\}.
$$
Here
$$
\widetilde{q}(x,\xi) = Mx \cdot \xi,
$$
where $M$ is a complex $n\times n$ block--diagonal matrix, each block being a Jordan one. Furthermore, the eigenvalues of $M$ are precisely those of
$2F$ in the upper half-plane, and the quadratic weight function $\Phi_1(x)$ satisfies,
$$
\Phi_1(x)\sim \abs{x}^2,\quad x\in \comp^n.
$$
The unitary equivalence between the operators ${\rm Op}^w(q)$ and ${\rm Op}^w(\widetilde{q})$ is realized by an operator of the form
$U_1 \circ T\circ U_2:L^2(\real^n) \rightarrow H_{\Phi_1}(\comp^n)$, where $U_2$ is a unitary metaplectic operator on $L^2(\real^n)$, while $T$ is an
FBI--Bargmann transform of the form {\rm (\ref{eq3.9})}, {\rm (\ref{eq3.10})}, and $U_1$ is an operator of the form {\rm (\ref{eq3.21})}.
\end{prop}

\bigskip
\noindent
Let us recall now that $\lambda_1,\ldots \,\lambda_n$ are the eigenvalues of the fundamental matrix $F$ of $q$ in the upper half-plane, repeated according to
their multiplicity. It follows from Proposition 3.1 that the action of the operator ${\rm Op}^w(\widetilde{q})$ on $H_{\Phi_1}(\comp^n)$ is given by
\begeq
\label{eq3.23}
{\rm Op}^w(\widetilde{q}) = \sum_{j=1}^n 2\lambda_j x_j D_{x_j} + \frac{1}{i}\sum_{j=1}^n \lambda_j + \sum_{j=1}^{n-1} \gamma_j x_{j+1} D_{x_j},\quad
\gamma_j \in \{0,1\}.
\endeq

\bigskip
\noindent
Let
$$
\Phi(x)=(1/2) \abs{x}^2,\quad x\in \comp^n,
$$
be the standard radial weight function. We observe that the weighted spaces $H_{\Phi_1}(\comp^n)$ and $H_{\Phi}(\comp^n)$ contain a common dense subset,
namely the space of holomorphic polynomials on $\comp^n$. Indeed, it is well known that the normalized monomials form an orthonormal basis in
$H_{\Phi}(\comp^n)$, while the density of holomorphic polynomials in $H_{\Phi_1}(\comp^n)$ has been explained in~\cite{HiSjVi}, and follows from the fact that
the finite linear combinations of the generalized eigenfunctions of ${\rm Op}^w(q)$ are dense in $L^2(\real^n)$. When restricted to the space of
holomorphic polynomials, the action of the operator ${\rm Op}^w(\widetilde{q})$ on $H_{\Phi_1}(\comp^n)$ is clearly similar to the action of the
closed densely defined operator ${\rm Op}^w(\widetilde{q})$ on the space $H_{\Phi}(\comp^n)$, the corresponding unbounded densely defined similarity transformation
$S: H_{\Phi_1}(\comp^n) \rightarrow H_{\Phi}(\comp^n)$ being the identity map, with ${\cal D}(S)$ being the space of holomorphic polynomials on $\comp^n$.

\medskip
\noindent
Assume now that the ${\cal PT}$--symmetry condition, $\overline{q(\kappa(x),-\kappa^t(\xi))} = q(x,\xi)$, $(x,\xi)\in \real^{2n}$, is maintained, and
that the spectrum of the quadratic operator ${\rm Op}^w(q)$ on $L^2(\real^n)$ is purely real. According to Remark 2.2, we then know that the eigenvalues
$\lambda_j$, $1\leq j \leq n$, are purely imaginary. If furthermore the nilpotent part in the Jordan decomposition of $F$ vanishes, then (\ref{eq3.23})
shows that
\begeq
\label{eq3.24}
{\rm Op}^w(\widetilde{q}) = \sum_{j=1}^n 2\lambda_j x_j D_{x_j} + \frac{1}{i}\sum_{j=1}^n \lambda_j,
\endeq
which, when equipped with the domain,
$$
{\cal D}({\rm Op}^w(\widetilde{q})) = \{u\in H_{\Phi}(\comp^n); (1+\abs{x}^2)u\in L^2_{\Phi}(\comp^n)\},
$$
is seen to be self-adjoint in the Bargmann space $H_{\Phi}(\comp^n)$, by means of a direct computation. This establishes the sufficiency part in Theorem
\ref{mth}.

\bigskip
\noindent
It only remains to verify the necessity in Theorem \ref{mth}. To this end, assume that the ${\cal PT}$--symmetric elliptic quadratic operator
${\rm Op}^w(q)$, $\Re q>0$, is similar to a self-adjoint operator, in the sense of the discussion in Section 1. It follows that for each
$\lambda \in {\rm Spec}({\rm Op}^w(q))$, the corresponding spectral subspace $E_{\lambda}\subset {\cal S}(\real^n)$ consists entirely of eigenvectors. We conclude
then from Proposition 3.1 that the nilpotent part in the Jordan decomposition of $F$ vanishes. The proof of Theorem \ref{mth} is complete.

\medskip
\noindent
{\bf Remark 3.3}. If we drop the ${\cal PT}$--symmetry assumption and merely assume that there are no Jordan blocks in the Jordan normal form of
the fundamental matrix $F$, then it follows from (\ref{eq3.24}) that the operator ${\rm Op}^w(\widetilde{q})$, acting on $H_{\Phi}(\comp^n)$, is normal,
and therefore, the original quadratic operator ${\rm Op}^w({q})$ acting on $L^2(\real^n)$ is similar to a normal operator.

\section{Example: quan\-tum har\-monic oscil\-lator\- with qu\-ad\-ra\-tic com\-plex interaction}
\setcounter{equation}{0}
The purpose of this section is to illustrate Theorem \ref{mth} by applying it to the following two-dimensional quadratic Schr\"odinger operator,
\begeq
\label{eq4.1}
{\rm Op}^w(q)= D_{x_1}^2 + D_{x_2}^2 + \omega_1^2 x_1^2 + \omega_2^2 x_2^2 + 2ig x_1 x_2.
\endeq
Here $\omega_j > 0$, $j=1,2$, $\omega_1 \neq \omega_2$, and $g\in \real$. The operator ${\rm Op}^w(q)$ has been considered in the work~\cite{CIN} as a
quantum model of a non-isotropic two-dimensional harmonic oscillator,  perturbed by an additional quadratic interaction, with a purely imaginary coupling constant.

\medskip
\noindent
The quadratic operator ${\rm Op}^w(q)$ is globally elliptic and $\P\T$--symmetric, with the corresponding involution given by $\kappa(x_1,x_2) = (-x_1,x_2)$.
In~\cite{CIN}, the authors show that a certain method of separation of variables fails to work for ${\rm Op}^w(q)$ precisely when
\begeq
\label{eq4.3}
2g = \pm (\omega_1^2 - \omega_2^2).
\endeq
Furthermore, in the case when (\ref{eq4.3}) holds, it is shown in~\cite{CIN} that the eigenfunctions of ${\rm Op}^w(q)$ do not form a complete set in $L^2(\real^2)$.

\medskip
\noindent
By applying Theorem 1.2, here we shall show the following result.

\begin{prop}
The spectrum of ${\rm Op}^w(q)$ is real precisely when
\begeq
\label{eq_cond1}
-\abs{\omega_1^2 - \omega_2^2} \leq 2g \leq \abs{\omega_1^2 - \omega_2^2},
\endeq
and ${\rm Op}^w(q)$ is similar to a self-adjoint operator if and only if
\begeq
\label{eq_cond2}
-\abs{\omega_1^2 - \omega_2^2} < 2g < \abs{\omega_1^2 - \omega_2^2}.
\endeq
\end{prop}
\begin{proof}
The Weyl symbol $q$ of the operator ${\rm Op}^w(q)$ is
\begeq
\label{eq4.31}
q(x,\xi) = \sum_{j=1}^2 \left(\xi_j^2 + \omega_j^2 x_j^2\right) + 2igx_1 x_2,
\endeq
and the corresponding fundamental matrix $F$ is given by
\begeq
\label{eq4.4}
F=\frac{1}{2} \left(\begin{array}{ll} q''_{\xi,x} & q''_{\xi,\xi}\\ -q''_{x,x} & -q''_{x,\xi} \end{array}\right)
=\left(\begin{array}{llll} 0 & 0 & 1 & 0 \\ 0 & 0 & 0 & 1 \\ -\omega_1^2 & -ig & 0 & 0 \\ -ig & -\omega_2^2 & 0 & 0 \end{array}\right).
\endeq
A straightforward computation then shows that
\begeq
\label{eq4.5}
{\rm det}\,(F-\lambda I) = \lambda^4 + (\omega_1^2 + \omega_2^2)\lambda^2 + \omega_1^2 \omega_2^2 + g^2,
\endeq
with the four eigenvalues $\lambda$ being given by
\begeq
\label{eq4.6}
\lambda^2 = -\frac{\omega_1^2 + \omega_2^2}{2} \pm \sqrt{\left(\frac{\omega_1^2-\omega_2^2}{2}\right)^2-g^2}.
\endeq
The expression under the square root sign is strictly less than $((\omega_1^2 + \omega_2^2)/2)^2$, so that the eigenvalues are non-vanishing, as we already know
since $q$ is elliptic. It vanishes precisely when (\ref{eq4.3}) holds, and from (\ref{eq4.6}) we conclude that the eigenvalues are simple precisely when
(\ref{eq4.3}) does not hold. When (\ref{eq4.3}) holds, then the spectrum of $F$ consists of two double eigenvalues,
\begeq
\label{eq4.7}
\lambda_1=\lambda_2 = i\left(\frac{\omega_1^2+\omega_2^2}{2}\right)^{1/2}, \quad \lambda_3=\lambda_4= -i\left(\frac{\omega_1^2+\omega_2^2}{2}\right)^{1/2}.
\endeq

\medskip
\noindent
When (\ref{eq_cond1}) holds, we see that the eigenvalues $\lambda$ of $F$ are on the imaginary axis, so thanks to (\ref{eq2.1}), we know that the spectrum of
${\rm Op}^w (q)$ is real. When (\ref{eq_cond1}) does not hold, the square root in (\ref{eq4.6}) is non-real, so the eigenvalues $\lambda$ are off the
imaginary axis, and ${\rm Op}^w(q)$ has spectrum away from $\real$, as was observed in Remark 2.2.

\medskip
\noindent
From the discussion so far, we get the first statement in the proposition concerning the reality of the spectrum of ${\rm Op}^w(q)$, and also the fact that
(\ref{eq_cond2}) is a sufficient condition for ${\rm Op}^w(q)$ to be similar to a self-adjoint operator. It only remains therefore to show that
${\rm Op}^w(q)$ is not similar to a self-adjoint operator when $2g=\pm (\omega_1^2-\omega_2^2)$, and in view of Theorem \ref{mth}, it suffices to show that Jordan
blocks do occur in $F$ then. In this case, we have two eigenvalues,
$$
\lambda_{\pm} = \pm i\left(\frac{\omega_1^2+\omega_2^2}{2}\right)^{1/2},
$$
each of algebraic multiplicity two. If no Jordan block is present in the spectral subspace of $\lambda_+$, say, then the rank of $F-\lambda_+ I$ is equal to two.
On the other hand, we have, using (\ref{eq4.4}),
\begeq
\label{eq4.8}
F-\lambda_+ I =\left(\begin{array}{llll} -\lambda_+ & 0 & 1 & 0 \\ 0 & -\lambda_+ & 0 & 1 \\ -\omega_1^2 & -ig & -\lambda_+ & 0
\\ -ig & -\omega_2^2 & 0 & -\lambda_+\end{array}\right),
\endeq
and the last three columns are linearly independent, since,
$$
{\rm det} \left(\begin{array}{lll} -\lambda_+ & 0 & 1 \\ -ig & -\lambda_+ & 0\\ -\omega_2^2& 0 & -\lambda_+\end{array}\right) =
-\lambda_+(\lambda_+^2 + \omega_2^2)\neq 0.
$$
The proof is complete.
\end{proof}

\begin{appendix}
\section{Positivity}
\label{pos}
\setcounter{equation}{0}
The purpose of this appendix is to provide a brief but complete account of the relevant aspects of the theory of positive complex Lagrangian planes,
required in the proof of Theorem \ref{mth}. See also Chapter 11 of~\cite{Sj82}, where a more general discussion is given.

\medskip
\noindent
Let $\Sigma \subset {\bf C}^{2n}$ be a real subspace of dimension $2n$ which is symplectic in the sense that
${{\sigma}_\vert}_{\Sigma }$ is a real and nondegenerate 2-form, where $\sigma =\sum_1^n d\xi _j\wedge dx_j$ is the complex symplectic 2-form
on ${\bf C}^{2n}={\bf C}_x^n\times {\bf C}_\xi ^n$. Then $\Sigma $ is maximally totally real in the sense that $\Sigma \cap i(\Sigma )=0$ and the real
dimension $\mathrm{dim}_{{\bf R}}(\Sigma )$ is maximal ($2n$).

\medskip
\noindent
Let $\iota =\iota _\Sigma :{\bf C}^{2n}\to {\bf C}^{2n}$ be the unique
anti-linear map which is equal to the identity on $\Sigma $. Clearly,
$$
\iota ^*\sigma =\overline{\sigma }.
$$

\medskip\par\noindent {\bf Example A.1.}
$\Sigma ={\bf R}^{2n}$, $\iota
(\rho )=\overline{\rho }$ is the usual complex conjugation.

\medskip\par\noindent {\bf Example A.2.}
Let $\Phi =\Phi (x)$ be a real
quadratic form on ${\bf C}^n$ such that the Levi matrix $\partial
_{\overline{x}}\partial _x\Phi $ is nondegenerate. Put
\begin{equation}\label{-1}
\Sigma =\Lambda _\Phi =\left\{ \left(x,\frac{2}{i}\frac{\partial \Phi }{\partial
x}(x)\right);\, x\in {\bf C}^n\right\}.
\end{equation}
Then $\Sigma $ is symplectic. In fact, using $x$ to parametrize
$\Sigma $ we get by a straightforward computation
$$
{{\sigma }_\vert}_{\Lambda _\Phi }=\sum_{j=1}^n \sum_{k=1}^n
\frac{2}{i}\frac{\partial ^2\Phi }{\partial \overline{x}_j\partial
  x_k}d\overline{x}_j\wedge dx_k.
$$
Using only that $\Phi $ is real, we see that this restriction is real. Hence $\Sigma $ is an I-Lagrangian manifold, i.e. a Lagrangian
manifold for the real symplectic form $\Im \sigma $. When the Levi form is nondegenerate, we see from the above expression that
${{\sigma}_\vert}_{\Sigma }$ is nondegenerate.

\bigskip
\noindent
Let $\Psi (x,y)$ be the unique holomorphic quadratic form on ${\bf C}_{x,y}^{2n}$ such that
$$
\Psi (x,\overline{x})=\Phi (x),\ x\in {\bf C}^n.
$$
Differentiating this relation, we get
$$
\partial _x^2\Psi =\partial _x^2\Phi ,\ \partial _x\partial _y\Psi
=\partial _x\partial _{\overline{x}}\Phi ,\ \partial _y^2\Psi
=\partial _{\overline{x}}^2\Phi ,
$$
keeping in mind that all second derivatives of $\Psi $, $\Phi$ are constant.

\medskip
The involution $\iota =\iota _{\Lambda _\Phi }$ is given by
$$
\iota :\ \left(y,\frac{2}{i}\overline{\frac{\partial \Psi }{\partial
    y}(x,\overline{y})}\right)\mapsto \left(x,\frac{2}{i}\frac{\partial \Psi
}{\partial x}(x,\overline{y})\right),
$$
or more explicitly by
\begin{equation}\label{0}
\iota :\ \left(y,\frac{2}{i}(\Phi ''_{x,\overline{x}}\overline{x}+\Phi
''_{x,x}y)\right)\mapsto \left(x,\frac{2}{i}(\Phi ''_{x,x}x+\Phi ''_{x,\overline{x}}\overline{y})\right).
\end{equation}

\medskip
\par
Let $\Lambda \subset {\bf C}^{2n}$ be a ${\bf C}$-Lagrangian space, i.e. a complex $n$-dimensional subspace such
that ${{\sigma }_\vert}_{\Lambda }=0 $. If $\Sigma \subset {\bf C}^{2n}$ is a real symplectic subspace of maximal dimension with
associated involution $\iota =\iota _\Lambda $, we consider the Hermitian bilinear form
\begin{equation}
\label{1}
b(\rho ,\mu )=\frac{1}{i}\sigma (\rho ,\iota (\mu )) \hbox{ on}\quad \Lambda \times \Lambda.
\end{equation}
The form (\ref{1}) was introduced by L. H\"ormander~\cite{H_EnsMath}, in the case when $\Sigma  = \real^{2n}$.

\medskip\par\noindent {\bf Proposition A.3.}
{\it The form {\rm (\ref{1})} is nondegenerate if and only if $\Lambda \cap \Sigma =\{0\}$.}

\begin{proof}
If $0\ne \rho \in \Lambda $ belongs to the radical of $b$, then $\iota (\rho )$ is symplectically orthogonal to $\Lambda $, and consequently,
$\iota (\rho )\in \Lambda$. The vectors $\frac{1}{2}(\rho +\iota (\rho ))$ and $\frac{1}{2i}(\rho -\iota (\rho ))$ belong to $\Sigma \cap \Lambda $
and at least one of them is $\ne 0$, so $\Sigma \cap \Lambda \ne \{0\}$. Conversely $\Sigma \cap \Lambda $ is contained in the radical of
$b$, so we have shown that the radical of $b$ is non-zero precisely when $\Sigma \cap \Lambda \ne \{0\}$.
\end{proof}

\medskip\par\noindent {\bf Example A.4.} If $\Sigma ={\bf R}^{2n}$ and $\Lambda $ is transversal to the fiber
$F=\{ (0,\xi );\, \xi \in {\bf C}^n\}$, then $\Lambda =\Lambda _\phi $: $\xi =\phi '(x)=\phi ''x$, where $\phi $ is a complex holomorphic quadratic form on
${\bf C}^n$. For $(x,\phi '(x))=(x,\phi ''x)\in \Lambda_\phi $ we get
$$
\frac{1}{2i}\sigma (\rho ,\overline{\rho })=(\Im \phi '')x \cdot \overline{x},
$$
where
$$
\Im \phi ''=\frac{1}{2i}(\phi ''-(\phi '')^*),
$$
so the signature of the form (\ref{1}) is equal to that of $\Im \phi''$.

\medskip
\par\noindent {\bf Definition A.5.} Let $\Lambda \subset {\bf C}^{2n}$
be a ${\bf C}$-Lagrangian space and let $\Sigma \subset {\bf C}^{2n}$
be a maximal real symplectic subspace with the corresponding
involution $\iota $. We say that $\Lambda $ is $\Sigma $-positive (or $\Sigma $-negative), if the Hermitian form (\ref{1}) is
positive (respectively negative) definite.

\medskip\par\noindent {\bf Example A.6.} If $\Sigma ={\bf R}^{2n}$ then $\Lambda $ is $\Sigma $-positive if and only if
$\Lambda=\Lambda _\phi $, where $\Im \phi ''>0$.

\begin{proof}
If $\Lambda =\Lambda _\phi $ we already know by Example A.4, that $\Sigma $-positivity of $\Lambda $ is
equivalent to that of $\Im \phi ''$. This proves the if-part of the
proposition. Conversely, if $\Lambda $ is $\Sigma $-positive, then we see that $\Lambda $ is transversal to the fiber $F$, so $\Lambda
=\Lambda _\phi $ and Example A.4 applies again.
\end{proof}

\medskip\noindent
The following result is a corollary of Proposition A.3 and Example A.4.

\medskip\par\noindent {\bf Corollary A.7} {\it Let $\Sigma $ be a fixed
maximal real symplectic subspace of ${\bf C}^{2n}$. Then the set of all
$\Sigma $-positive ${\bf C}$-Lagrangian spaces is a connected
component in the set of all ${\bf C}$-Lagrangian spaces that are
transversal to $\Sigma $.}

\medskip
\noindent
In fact, after applying a suitable complex canonical transformation, we may assume that $\Sigma ={\bf R}^{2n}$ and we see
from Example A.6 that the set of all $\Sigma $-positive ${\bf C}$-Lagrangian spaces is connected. Proposition A.3 then shows that it
is a connected component in the set of all ${\bf C}$-Lagrangian spaces that are transversal to $\Sigma$.

\bigskip
\noindent
Now, let $\Sigma =\Lambda _\Phi $ be as in (\ref{-1}), where
\begin{equation}\label{3}
\partial _{\overline{x}}\partial _x\Phi >0.
\end{equation}

\medskip\par\noindent {\bf Proposition A.8.} {\it The fiber $F=\{(0,\eta   );\, \eta \in {\bf C}^n\}$ is $\Sigma $-negative.}

\begin{proof} Using (\ref{0}) we see that $\iota
(0,\eta )=(x,\xi )$ is given by
\begin{equation}\label{4}
\xi =\frac{2}{i}\Phi ''_{x,x}x,\ x=\frac{1}{2i}(\Phi
''_{\overline{x},x})^{-1}\overline{\eta }.
\end{equation}
It follows that
$$
\frac{1}{i}\sigma (0,\eta ;x,\xi )=-\frac{1}{2} \eta \cdot (\Phi''_{\overline{x},x})^{-1}\overline{\eta} \le -\frac{1}{C}|\eta |^2.
$$
\end{proof}

\medskip\par
As we saw in (\ref{4}), the $\Sigma $-positive space $\iota (F)$ is given by
\begin{equation}
\label{5}
\xi =\frac{2}{i}\Phi ''_{x,x}x=\frac{\partial }{\partial x}\left(\frac{2}{i} \Phi ''_{x,x}x \cdot x \right).
\end{equation}
Here, we notice that $\Phi (x)=\Phi_\mathrm{plh}(x)+\Phi_\mathrm{herm}(x)$, where
$$
\Phi _\mathrm{plh}(x)=\frac{1}{2}\Phi''_{x,x}x \cdot x + \frac{1}{2}\Phi''_{\overline{x},\overline{x}}\overline{x} \cdot \overline{x} =
\frac{1}{2}(\Phi (x)-\Phi (ix))
$$
is the pluriharmonic (plh) part and
$$
\Phi _\mathrm{herm}(x) = \Phi ''_{\overline{x},x}x \cdot \overline{x} = \frac{1}{2}(\Phi (x)+\Phi (ix))
$$
is the (positive definite) Hermitian part. We conclude that the positive ${\bf C}$-Lagrangian space $\iota (F)$ is of the form $\iota
(F)=\Lambda _{\Phi _\mathrm{plh}}$, where $\Phi(x) -\Phi_\mathrm{plh}(x)\asymp |x|^2$.

\medskip\par\noindent {\bf Proposition A.9.} {\it Let $\Sigma =\Lambda_\Phi $ be as in {\rm (\ref{-1})}, {\rm (\ref{3})}. A ${\bf C}$-Lagrangian
space is $\Sigma $-positive if and only if $\Lambda =\Lambda _{\widetilde{\Phi}}$, where $\widetilde{\Phi }$ is pluriharmonic and
$\Phi -\widetilde{\Phi}\asymp |x|^2$. }

\begin{proof} If we decompose the pluriharmonic form $\widetilde{\Phi }$ as
$\widetilde{\Phi }=\Phi _\mathrm{plh}+\widehat{\Phi }$, we see that $\Phi
-\widetilde{\Phi }> 0$ precisely when $|\widehat{\Phi }(x)|<\Phi_\mathrm{herm}(x)$, $x\ne 0$. Consequently, the set
$\{ \Lambda_{\widetilde{\Phi }};\, \widetilde{\Phi }\hbox{ is plh and }\Phi-\widetilde{\Phi} >0\}$ is a connected component of the set of all
${\bf C}$-Lagrangian spaces that are transversal to $\Lambda_\Phi $. It contains $\Lambda_{\Phi _\mathrm{plh}}$ which is $\Lambda_\Phi $-positive,
so by Corollary A.7 it is equal to the set of all ${\bf C}$-Lagrangian spaces, which are $\Lambda_\Phi $-positive.
\end{proof}

\section{Quadratic Fourier integral opera\-tors in the com\-p\-lex do\-ma\-in}
\label{FIO}
The Fourier integral operators encountered in this appendix arise when quantizing complex linear canonical transformations. The following discussion can
therefore be viewed as a linear version of the theory presented in Chapters 3 and 4 in~\cite{Sj82}. See also Chapter 3 of~\cite{Lebeau}.

\medskip\par
We start with the formal theory. A Fourier integral operator is an operator of the form
\begin{equation}\label{8}
Au(x)=\iint e^{i\phi (x,y,\theta )}au(y)dyd\theta ,
\end{equation}
where $a\in {\bf C}$ and $\phi $ is a holomorphic quadratic form on
${\bf C}^{n+n+N}_{x,y,\theta }$. We assume that $\phi $ is a nondegenerate phase function in the sense of~\cite{HoFIO}, so that
\begin{equation}
\label{9}
d\frac{\partial \phi }{\partial \theta _1},...,d\frac{\partial \phi}{\partial \theta _N}\hbox{ are linearly independent,}
\end{equation}
or equivalently that
\begin{equation}
\label{10}\phi ''_{\theta ,(x,y,\theta )}\hbox{ is of rank }N.
\end{equation}

\medskip
\noindent
 We may assume without loss of generality that
\begin{equation}
\label{11}\phi ''_{\theta ,\theta }=0.
\end{equation}
Indeed, if $\phi ''_{\theta ,\theta }\ne 0$, we may assume after a
linear change of coordinates in $\theta $, that $\theta =(\theta
',\theta '')\in {\bf C}^{N-d}\times {\bf C}^d$, $\phi ''_{\theta
  ',\theta }=0$ and that $\det \phi ''_{\theta '',\theta ''}\ne
0$. Then $\theta ''\mapsto \phi (x,y,\theta )$ has a unique critical
point $\theta ''_{c}(x,y)$ which is nondegenerate and integrating out
the $\theta ''$-variables in (\ref{8}), we get formally
\begin{equation}\label{12}Au(x)=\iint e^{i\psi (x,y,\theta ' )}bu(y)dyd\theta' ,\end{equation}
where $b$ is a constant non-vanishing multiple of $a$ and $\psi
(x,y,\theta ')=\phi (x,y,\theta ',\theta ''_c)$ is a nondegenerate
phase function with $\psi ''_{\theta ',\theta '}=0$.

\medskip
\noindent
Assuming (\ref{11}) until further notice, we see that (\ref{10}) becomes
\begin{equation}\label{13}
\mathrm{rank\,}(\phi ''_{\theta ,(x,y)})=N,
\end{equation}
and in particular, $N\le 2n$.

\bigskip
\noindent
Let
\begin{equation}
\label{14}
C_\phi =\{ (x,y,\theta )\in {\bf C}^{2n+N};\, \phi '_\theta (x,y,\theta )=0\}
\end{equation}
This is a subspace of dimension $2n$ and we consider the corresponding
canonical relation
\begin{equation}\label{15}
\kappa :\, (y,\eta )\mapsto (x,\xi ),
\end{equation}
defined by its graph
\begin{equation}\label{16}
\mathrm{graph\,}(\kappa )=\{ (x,\xi ;y,\eta )=(x,\phi '_x(x,y,\theta
);y,-\phi '_y(x,y,\theta ));\, (x,y,\theta )\in C_\phi \}.
\end{equation}
It is easy to check that $\mathrm{dim\, }(\mathrm{graph\,}(\kappa
))=2n$ and that
$$
{{(\sigma _{x,\xi }-\sigma _{y,\eta })}_\vert}_{\mathrm{graph\,}(\kappa
)}=0,
$$
where $\sigma _{x,\xi }-\sigma _{y,\eta }=\sum_1^nd\xi _j\wedge
dx_j-\sum_1^nd\eta  _j\wedge dy_j$.

\bigskip
\noindent
Assume that
\begin{equation}\label{17}
\kappa \hbox{ is a canonical transformation},
\end{equation}
which means that the maps $\mathrm{graph\,}(\kappa )\ni (x,\xi ;y,\eta)\mapsto (x,\xi )\in {\bf C}^{2n}$ and $\mathrm{graph\,}(\kappa )\ni (x,\xi ;y,\eta
)\mapsto (y,\eta )\in {\bf C}^{2n}$ are bijective. Actually, the bijectivity of one of the maps implies that of the other. In fact, the
bijectivity of the first map is equivalent to that of $C_\phi \ni (x,y,\theta )\mapsto (x,\phi '_x(x,y,\theta ))\in {\bf C}^{2n}$ which
is equivalent to that of
$$
{\bf C}^{2n+N}\ni (x,y,\theta )\mapsto (x,\phi '_x,\phi '_\theta )=
(x,\phi ''_{x,x}x+\phi ''_{x,y}y+\phi ''_{x,\theta }\theta ,\phi
''_{\theta ,x}x+\phi ''_{\theta ,y}y)\in {\bf C}^{2n+N}.
$$
This is equivalent to the bijectivity of
$$
\begin{pmatrix}\phi ''_{x,y} &\phi ''_{x,\theta }\\
\phi ''_{\theta ,y} &0\end{pmatrix}.
$$
and implies that $\phi ''_{\theta ,y}$ and $\phi ''_{x,\theta }$ are of maximal rank $N$ and in particular that $N\le n$.

\medskip
\noindent
We can write
\begin{equation}\label{18}
\phi (x,y,\theta )=g(x,y)+\sum_{j=1}^N f_j(x,y)\theta _j,
\end{equation}
where $g$ is quadratic and $f_j$ are linear. Further, $d_xf_1,...,d_xf_N$ are linearly independent and similarly for the
$y$-differentials. The subspace $C_\phi $ is given by the equations $f_j(x,y)=0$, $j=1,..,N$ in ${\bf C}^{2n+N}$ and the same equations in the $x,y$ space
define the projection of $C_\phi $, or equivalently, the projection $\pi _{x,y}(\mathrm{graph\,}(\kappa))$, so $N$ is uniquely determined by $\kappa $. Let
\begin{equation}\label{19}
\widetilde{\phi} (x,y,\widetilde{\theta })=\widetilde{g}(x,y)+\sum_{j=1}^N
  \widetilde{f}_j(x,y)\widetilde{\theta }_j
\end{equation}
be a second nondegenerate phase function that generates the same canonical
transformation $\kappa $. Then $f_1=...=f_N=0$ and
$\widetilde{f}_1=...=\widetilde{f}_N=0$ define the same subspace of
${\bf C}^{2n}$ and after a linear change
of the $\widetilde{\theta }$ coordinates, we may assume that
$$
\widetilde{\phi }=\widetilde{\phi }(x,y,\theta
)=\widetilde{g}(x,y)+\sum_{j=1}^Nf_j(x,y)\theta _j .
$$
The fact that the two phases give rise to the same canonical
transformation now means that there are linear forms $k_j(x,y)$,
$1\le j\le N$ such that
\begin{equation}\label{21}
\widetilde{g}(x,y)=g(x,y)+\sum_{j=1}^N f_j(x,y)k_j(x,y).
\end{equation}
Then
\begin{equation}\label{22}\widetilde{\phi }(x,y,\theta )=\phi
  (x,y,\theta +k(x,y)),\end{equation}
which will imply that $\widetilde{\phi }$ gives rise to the same Fourier integral operators, once we have explained how to choose the
contour of integration in (\ref{8}).

\bigskip
\noindent
Let $\Phi _0$ and $\Phi _1$ be real quadratic forms on (different copies of) ${\bf C}^n$.

\begin{prop} Let $\phi $ be a nondegenerate phase function as above, with the associated canonical transformation $\kappa $.
The following statements are equivalent:
\begin{itemize}
\item[1)] $\kappa (\Lambda _{\Phi _0})=\Lambda _{\Phi _1}$.
\item[2)] The quadratic form
\begin{equation}\label{23}(y,\theta )\mapsto \Phi _0(y)-\Im \phi (0,y,\theta
  )\end{equation}
is nondegenerate so that
\begin{equation}\label{24}
(y,\theta )\mapsto \Phi _0(y)-\Im \phi (x,y,\theta )
\end{equation}
has a unique critical point $(y(x),\theta (x))$ for every $x\in {\bf
  C}^n$. Moreover,
$$
\Phi _1(x)=\mathrm{vc}_{y,\theta }(\Phi _0(y)-\Im \phi (x,y,\theta )).
$$\rm
\end{itemize}
\end{prop}
\begin{proof}
We first assume 2). Then at the critical point $(y(x),\theta (x))$, we have
\[
\begin{split}
&\frac{2}{i}\frac{\partial \Phi _0}{\partial
  y}=-\frac{2}{i}\frac{\partial }{\partial y}(-\Im \phi (x,y,\theta
))=-\frac{\partial }{\partial y}\phi (x,y,\theta ),\\
&\frac{\partial }{\partial \theta }\phi (x,y,\theta )=0,\\
&\frac{2}{i}\frac{\partial \Phi _1}{\partial
  x}=\frac{2}{i}\frac{\partial }{\partial x}(-\Im \phi )(x,y,\theta
)=\frac{\partial \phi }{\partial x}(x,y,\theta ),
\end{split}
\]
which means that 1) holds.

\par Now assume 1). Then for every $x\in {\bf C}^n$ there exists a
unique $y\in {\bf C}^n$ such that $x\in \pi _x(\kappa
(y,\frac{2}{i}\frac{\partial \Phi _0}{\partial y}(x)))$. Equivalently
for every $x\in {\bf C}^n$, (\ref{24}) has a unique critical point
$(y(x),\theta (x))$. Since we are dealing with second order polynomials
this critical point is nondegenerate. If $\widetilde{\Phi }_1(x)$ is
the critical value, we also see that
$$
\frac{2}{i}\frac{\partial \Phi _1}{\partial x}=\frac{2}{i}\frac{\partial \widetilde{\Phi }_1}{\partial x}
$$
and since $\Phi _1$ and $\widetilde{\Phi }_1$ are quadratic forms, we
conclude that $\widetilde{\Phi }_1=\Phi _1$.
\end{proof}

\bigskip
\noindent
If $\kappa (\Lambda _{\Phi _0})=\Lambda _{\Phi _1}$ and $\Phi _0$ is pluriharmonic (so that $\Lambda _{\Phi _0}$ is a
${\bf C}$-Lagrangian space), then $\Phi _1$ is also pluriharmonic and the nondegenerate quadratic form (\ref{23}) has the signature $(n+N,n+N)$. In
general, when $\Phi _0(y)$ and hence $\Phi _0(y)-\Im \phi (0,y,\theta)$ is plurisubharmonic, then (\ref{23}) has at most $n+N$
negative eigenvalues, and in order to find a suitable integration contour, when realizing $A$, we wish this
maximal number of negative eigenvalues to be attained. The following result gives a positive answer in the case that we are interested in.

\begin{prop} Assume that $\kappa(\Lambda _{\Phi _0})=\Lambda _{\Phi _1}$, where $\partial_{\overline{z}}\partial _z\Phi _0$ and $\partial
_{\overline{z}}\partial _z\Phi _1$ are positive definite. Then the signature of the quadratic form {\rm (\ref{23})} is equal to $(n+N,n+N)$.
\end{prop}
\begin{proof}
Let $p$ be a positive definite quadratic form on $\Lambda _{\Phi _0}$ and let $p$ also denote the
holomorphic extension to ${\bf C}^{2n}$ (which is a well defined holomorphic quadratic form since $\Lambda _{\Phi _0}$ is maximally
totally real). Then the spectrum of the fundamental matrix $F_p$ is given by the eigenvalues $\pm i\mu _1,...,\pm i\mu _n$, where $\mu
_j>0$. Let $\Lambda _{0,+\infty }$ denote the spectral subspace of ${\bf C}^{2n}$ corresponding to $i\mu _1,...,i\mu _n$. Then we know from~\cite{Sj74}
that $\Lambda_{0,+\infty}$ is positive with respect to $\Lambda_{\Phi _0}$, so that, according to Proposition A.9, $\Lambda _{0,+\infty }$ is of the form
$\xi = \frac{2}{i}\frac{\partial \Phi _{0,+\infty }}{\partial x}(x)$, where $\Phi _{0,+\infty }$ is a pluriharmonic quadratic form such that
$\Phi _0(x)-\Phi_{0,+\infty }(x)\asymp |x|^2$.

\medskip
\noindent
The spectral subspace $\Lambda_{0,+\infty}$ can be viewed as the unstable manifold for the $H_{-ip}$-flow and if we put
$\Lambda _{0,t}=\exp(-itH_p)(\Lambda_{\Phi _0})$, for $t\in \real$, then $\Lambda _{0,t}$ is I-Lagrangian
and R-symplectic, and $\Lambda _{0,t}$ converges to $\Lambda_{0,+\infty }$ exponentially fast. Furthermore, at least for small
$t>0$, we have $\Lambda _{0,t}=\Lambda _{\Phi _{0,t}}$, where $\Phi_{0,t}(x)$ is the real quadratic form obtained by solving the eikonal
equation
\begin{equation}
\label{25}
\frac{\partial }{\partial t}\Phi _{0,t}(x)+\Re p\left(x,\frac{2}{i}\frac{\partial }{\partial x}\Phi _{0,t}(x)\right)=0.
\end{equation}
See for instance,~\cite{SjHokk}, as well as Remark 11.7 in~\cite{HeSjSt}, for more details. Notice here that since $p$ is
constant along the $H_{-ip}$ trajectories, we have $\Re p=p$ in (\ref{25}) and
$p\left(x,\frac{2}{i}\frac{\partial }{\partial x}\Phi_{0,t}(x)\right)\asymp |x|^2$, non-uniformly with respect to $t$.

\medskip
\noindent
A priori it is not clear that we can solve (\ref{25}) for all $t\ge 0$. However, as long as the solution exists, we see that
$t\mapsto \Phi _{0,t}$ is a decreasing family of strictly plurisubharmonic quadratic forms, so that $\Phi _{0,t}(x)\le \Phi _0(x)$. For such
a $t$, we write $\psi := \Phi _{0,t}=\psi _{\mathrm{herm}}+\psi _\mathrm{plh}$, where
$\psi_{\mathrm{herm}} (x)=\frac{1}{2}(\psi (x)+\psi (ix))$ is the
Hermitian part and $\psi_{\mathrm{plh}} (x)=\frac{1}{2}(\psi (x)-\psi
(ix))$ is the pluriharmonic part. Similarly, we write $\phi := \Phi _0(x)=\phi_\mathrm{herm}+\phi _\mathrm{plh}$. From $\psi \le \phi $ we get
$0\le \psi _\mathrm{herm}\le \frac{1}{2}(\phi(x)+\phi (ix))= \phi _\mathrm{herm}$. Also,
\[
\begin{split}
\psi _\mathrm{plh}(x)\le (\phi _\mathrm{herm}(x)-\psi
_\mathrm{herm}(x))+\phi _\mathrm{plh}(x)\le \phi (x),\\
-\psi _\mathrm{plh}(x)=\psi_\mathrm{plh}(ix)\le \phi (ix),
\end{split}
\]
so that
$$
-\phi (ix)\le \psi _\mathrm{plh}(x)\le \phi (x).
$$
This means that we have $t$-independent upper and lower bounds on $\Phi _{0,t}$, which prevent an explosion in the eikonal
equation (\ref{25}). Consequently, $\Lambda _{0,t}=\Lambda _{\Phi _{0,t}}$ for all $0\le t<\infty$, and the exponential convergence of $\Lambda _{0,t}$
to $\Lambda _{0,+\infty }$ when $t\to +\infty $ shows that
$$
\Phi _{0,t}\to \Phi _{0,+\infty }\hbox{ exponentially fast, when }t\to
+\infty .
$$

\medskip
\noindent
Let us define $q$ on ${\bf C}^{2n}$ by $q\circ \kappa =p$. Replacing $(\Phi _0,p)$ in the above discussion by $(\Phi _1,q)$, we get a
decreasing family of strictly plurisubharmonic quadratic forms $\Phi _{1,t}$ for $0\le t<+\infty $, converging to the pluriharmonic
form $\Phi _{1,+\infty }$ when $t\to+\infty $, with the property that
$$
\kappa (\Lambda _{\Phi _{0,t}})=\Lambda _{\Phi _{1,t}},\ 0\le t\le +\infty .
$$
Proposition B.1 now shows that
$$
(y,\theta )\mapsto \Phi _{0,t}(y)-\Im \phi (0,y,\theta )
$$
is a nondegenerate quadratic form for $0\le t\le +\infty $, necessarily
with $t$-independent signature. When $t=+\infty $ the signature is
equal to $(n+N,n+N)$ so this is also the case when $t=0$, which gives
the proposition.
\end{proof}

\bigskip
\noindent
Now let us recall the representation (\ref{18}), where $\phi ''_{\theta,y}$ and $\phi ''_{\theta ,x}$ are of maximal rank $N$. After
separate linear changes of the $x$ and $y$ coordinates, we may assume, in order to simplify the notation only, that
\begin{equation}
\label{26}
\phi (x,y,\theta )=g(x,y)+\sum_{j=n-N+1}^N(x_j-y_j)\theta
_j=g(x,y)+ (x''-y'') \cdot \theta,
\end{equation}
writing $x=(x',x'')\in {\bf C}^{n-N}\times {\bf C}^N$ and similarly for $y$.
We see directly that
$$
-\Im \phi (0,0,y'',\theta )+\Phi _0(0,y'')=-\Im
g(0,0,y'')+\Phi _0(0,y'')+\Im \left(y'' \cdot \theta\right)
$$
is a nondegenerate quadratic form of signature $(N,N)$. More
generally, the second order polynomial
\begin{equation}\label{26.5}
(y'',\theta )\mapsto -\Im \phi (x,y,\theta )+\Phi _0(y)=-\Im
g(x,y)+\Phi _0(y)-\Im \left((x''-y'') \cdot \theta\right)
\end{equation}
has a unique critical point of signature $(2N,2N)$. The critical point $(y''_c(x,y'),\theta _c(x,y'))$ is given by $y''_c=x''$, $\theta
_c=\frac{2}{i}\frac{\partial \Phi _0}{\partial
  y''}(y',x'')+\frac{\partial g}{\partial y''}(x,y',x'')$ and the
critical value is
\begin{equation}\label{27}
-\Im \phi (x,y',x'',\theta _c)+\Phi _0(y',x'')=-\Im g(x,y',x'')+\Phi _0(y',x'').
\end{equation}

\medskip
\noindent
On the other hand, by Proposition B.2, we know that
\begin{equation}\label{28}
(y,\theta )\mapsto -\Im \phi (x,y,\theta )+\Phi _0(y)
\end{equation}
has a unique critical point which is of signature $(n+N,n+N)$. It follows that
\begin{equation}\label{29}
y'\mapsto -\Im g(x,y',x'')+\Phi _0(y',x'')
\end{equation}
has a unique critical point which is of of signature $(n-N,n-N)$, and the corresponding critical value is equal to that of (\ref{28}),
namely $\Phi _1(x)$.

\bigskip
\noindent
Let $\gamma \subset {\bf C}_{y'}^{n-N}$ be a real-linear subspace of dimension $n-N$ such that
\begin{equation}\label{30}
-\Im g(0,y',0)+\Phi _0(y',0)\asymp -|y'|^2,\ y'\in \gamma .
\end{equation}
Then
\begin{equation}\label{31}
-\Im g(x,y',x'')+\Phi _0(y',x'')-\Phi _1(x)\asymp -|y'-y'_c(x)|^2,\ y'\in y_c(x)+\gamma .
\end{equation}

\medskip
\noindent
Let us consider the contour $\Gamma (x)\subset {\bf C}_{y,\theta
}^{n+N}$ of real dimension $n+N$, given by
\[
y'\in y'_c(x)+\gamma ,\ y''\in {\bf C}^N,\ \theta =\theta _c(x,y')+iC\overline{(x''-y'') },
\]
parametrized by $(y',y'')\in \gamma \times {\bf C}^N$. Along $\Gamma
(x)$ we have
\begin{equation}\label{32}
-\Im \phi (x,y,\theta )+\Phi _0(y)-\Phi _1(x)\asymp
-\left(|y'-y'_c(x)|^2 + |y''-x''|^2 \right),
\end{equation}
provided that we choose $C>0$ large enough.

\medskip
\noindent
Let $u\in H_{\Phi _0}({\bf C}^n)$, and let us realize $Au=A_\Gamma u$ in (\ref{8}) by integrating over $\Gamma (x)$. We get
\begin{equation}\label{33}\begin{split}
e^{-\Phi _1(x)}A_\Gamma u(x)&=\int_{y''\in
  {\bf C}^N}\int_{y'\in y'_c(x)+\gamma } e^{i\psi _\Gamma
  (x,y)}\widetilde{a}(u(y)e^{-\Phi _0(y)})dy'L(dy'')\\&=e^{-\Phi
  _1}A_\Gamma e^{\Phi _0}(ue^{-\Phi _0}),\end{split}
\end{equation}
where $-\Im \psi _\Gamma (x,y)$ is equal to the right hand side in (\ref{32}). From (\ref{32}) we see that the integral in (\ref{33})
converges for each $x\in {\bf C}^n$, and by a standard application of Stokes' formula, we also know that if $\widetilde{\Gamma }$ is a second
contour with the same properties as $\Gamma $, then $A_{\Gamma}u(x)=A_{\widetilde{\Gamma }}u(x)$.

\medskip
\noindent
We now wish to estimate the ${\cal L}(L^2,L^2)$--norm of the effective operator $e^{-\Phi_1}A_\Gamma e^{\Phi _0}$ in (\ref{33}). We know that the map
\begin{equation}\label{34}
x\mapsto y_c(x)=(y_c'(x),x'')
\end{equation}
is bijective. Replacing $x$ by the new $x$ coordinates $(y'_c(x),x'')$, we may reduce ourselves to the case when
$y'_c(x)=x'$. After a simultaneous rotation in the $x'$ and $y'$ variables we may further assume that $\gamma ={\bf R}^{n-N}$. Then
from  (\ref{32}), (\ref{33}), we get
\begin{equation}\label{35}
\begin{split}
&e^{-\Phi _1(x)}A_\Gamma u(x)=\\
&\iint_{t'\in {\bf R}^{n-N}\atop y''\in
  {\bf C}^N}{\cal O}(1)e^{-\frac{1}{C}(|{\rm Re}\,x'-t'|^2+|x''-y''|^2)}e^{-\Phi _0(t'+i{\rm Im}\,x',y'')}u(t'+i\Im x',y'')dt'L(dy'').
\end{split}
\end{equation}
By Schur's lemma,
$$
\Vert e^{-\Phi _1}A_\Gamma u\Vert_{L^2(E_{t'})}\le {\cal O}(1)\Vert
e^{-\Phi _0}u\Vert_{L^2(E_{t'})},
$$
uniformly for $t'\in {\bf R}^{n-N}$, where $E_{t'}=\{ x\in {\bf C}^n;\, \Im x'=t'\}$. Integrating this with respect to $t'$, we get
$$
\Vert e^{-\Phi _1}A_\Gamma u\Vert_{L^2({\bf C}^n)}\le {\cal O}(1)\Vert
e^{-\Phi _0}u\Vert_{L^2({\bf C}^n)}.
$$
We have proved the following result.

\begin{prop} Realizing the operator $A$ by means of the contour $\Gamma$ above, we get a bounded operator
$A=A_\Gamma :H_{\Phi _0}\to H_{\Phi _1}$. We get the same operator if we replace $\Gamma $ by any other contour with the same properties.
\end{prop}

\medskip
\noindent
The next result can be proved by the general methods of~\cite{Sj82}.

\begin{prop} Let $(B, \widetilde{\kappa})$ have same general properties as $(A,\kappa )$ and assume that
$\widetilde{\kappa } (\Lambda _{\Phi _1})=\Lambda _{\Phi _2}$, where $\Phi_2$ is a strictly plurisubharmonic quadratic form. Let
$\widetilde{\Gamma } $ be a contour allowing to realize $B=B_{\widetilde{\Gamma }}:H_{\Phi _1}\to H_{\Phi _2}$. Then
there exists a Fourier integral operator $C$ as above with associated canonical transformation $\widetilde{\kappa }\circ \kappa $ and a
contour $\widehat{\Gamma}$ such that $C_{\widehat{\Gamma}}=B_{\widetilde{\Gamma }}\circ A_{\Gamma }$.
\end{prop}
\begin{proof} We shall only recall the general ideas of the proof. First of all, the elimination of superfluous $\theta$
coordinates above was only in order to simplify the presentation, and it is very easy to realize Fourier integral operators with more than
the minimal number of such fiber variables. Formally $B\circ A$ then becomes such a Fourier integral operator with $\widetilde{\kappa }\circ
\kappa $ as the associated canonical transformation. When composing $B_{\widetilde{\Gamma }}\circ A_\Gamma $ we automatically get a good contour, and as we
have pointed out, this contour can be deformed into any other good contour for $B\circ A$. Superfluous $\theta $ variables can then be eliminated
by the exact version of stationary phase.
\end{proof}

\medskip
\noindent
A special case is when $\kappa =\mathrm{id}$. We then get a multiple of the identity operator. If $\widetilde{\kappa }$ in the
last proposition is equal to $\kappa ^{-1}$, then we conclude that $B\circ A$ is a multiple of the identity, and assuming that the
amplitude $a$ in (\ref{8}) is $\ne 0$, we can choose the amplitude $b$ in a corresponding representation of $B$ in such a way that
$B\circ A=I$. In particular, we see that $A_\Gamma :H_{\Phi _0}\to H_{\Phi _1}$ in Proposition B.3 has a bounded (two-sided) inverse.

\end{appendix}


\begin{thebibliography}{40}

\bibitem{ACS}  A. A. Andrianov, F. Cannata, and A. V. Sokolov,
{\it Nonlinear supersymmetry for non--Hermitian, non--diagonalizable Hamiltonians. {\rm I}. General properties}, Nuclear Phys. {\bf B773} (2007), 107--136.

\bibitem{BaZn} F. Bagarello and M. Znojil, {\it Nonlinear pseudo-bosons versus hidden Hermiticity: {\rm II}. The case of unbounded operators}, J. Phys. A
{\bf 45} (2012), 115311 (13 pp).

\bibitem{Bar61} V. Bargmann, {\it On a Hilbert space of analytic functions and an associated integral transform}, Comm. Pure Appl. Math. {\bf 14} (1961),
187--214.

\bibitem{BB} C. M. Bender and S. Boettcher, {\it Real spectra in non-Hermitian Hamiltonians having ${\cal PT}$--symmetry},
Phys. Rev. Lett. {\bf 80} (1998), 5243–-5246.

\bibitem{BBM} C. M. Bender, S. Boettcher, and P. N. Meisinger, {\it ${\cal PT}$-symmetric quantum mechanics}, J. Math. Phys. {\bf 40} (1999), 2201-–2229.

\bibitem{Boutet} L. Boutet de Monvel, {\it Hypoelliptic operators with double characteristics and related pseudo-differential operators},
Comm. Pure Appl. Math. {\bf 27} (1974), 585-–639.

\bibitem{CGrSj} E. Caliceti, S. Graffi, and J. Sj\"ostrand,
{\it ${\cal PT}$--symmetric non-self-adjoint operators, diagonalizable and non-diagonalizable, with a real discrete spectrum},
J. Phys. A {\bf 40} (2007), 10155-–10170.

\bibitem{CIN} F. Cannata, M. V. Ioffe, and D. N. Nishnianidze, {\it Exactly solvable nonseparable
and nondiagonalizable two-dimensional model with quadratic complex interaction}, J. Math. Phys. {\bf 51} (2010), 022108, 14 pp.

\bibitem{Da1} E. B. Davies, {\it Pseudo-spectra, the harmonic oscillator and complex resonances},
R. Soc. Lond. Proc. Ser. A Math. Phys. Eng. Sci. {\bf 455} (1999), 585-–599.

\bibitem{DeSjZw} N. Dencker, J. Sj\"ostrand, and M. Zworski, {\it Pseudo-\-spectra of se\-mi\-clas\-si\-cal
(pseudo)\-diffe\-ren\-tial ope\-rators}, Comm. Pure Appl. Math. {\bf 57} (2004), 384--415.

\bibitem{DiSj} M. Dimassi and J. Sj\"ostrand, {\it Spectral asymptotics in the semi-classical limit}, Cambridge University Press, 1999.

\bibitem{HeSj84} B.~Helffer and J.~Sj\"ostrand, {\it Multiple wells in the semiclassical limit. {\rm I}.},
Comm. P.D.E. {\bf 9} (1984), 337-–408.

\bibitem{HeSjSt} F. H\'erau, J. Sj\"ostrand, and C. Stolk, {\it Semiclassical analysis for the Kramers-Fokker-Planck equation},
Comm. PDE, {\bf 30} (2005), 689--760.

\bibitem{HiPr1} M. Hitrik and K. Pravda-Starov {\it Spectra and semigroup smoothing for non-elliptic quadratic operators},
Math. Ann., {\bf 344} (2009), 801-–846.

\bibitem{HiSjVi} M. Hitrik, J. Sj\"ostrand, and J. Viola, {\it Resolvent estimates for elliptic quadratic differential operators}, Analysis and PDE, to appear.

\bibitem{H_EnsMath} L. H\"ormander, {\it On the existence and the regularity of solutions of linear pseudo-differential equations}, Ens. Math. {\bf 17}
(1971), 99–-163.

\bibitem{HoFIO} L. H\"ormander, {\it Fourier Integral Operators {\rm I}}, Acta Math. {\bf 127} (1971), 79--183.

\bibitem{H_book}
L. H\"{o}rmander, {\it The analysis of linear partial differential operators} (vol. I--IV), Springer Verlag (1985).

\bibitem{Ho95}
L. H\"{o}rmander, {\it Symplectic classification of quadratic forms, and general Mehler formulas}, Math. Z., {\bf 219} (1995), 413--449.

\bibitem{Lebeau} G. Lebeau, {\it Deuxi\`eme microlocalisation sur les sous-vari\'et\'es isotropes}, Ann. Inst. Fourier {\bf 35} (1985), 145-–216.
58G07 (35A27 46F15)

\bibitem{MeSj78} A. Menikoff and J. Sj\"ostrand, {\it On the eigenvalues of a class of hypoelliptic operators},  Math. Ann. {\bf 235} (1978), 55–-85.

\bibitem{Mo1} A. Mostafazadeh, {\it Pseudo-Hermiticity versus ${\cal PT}$--symmetry: the necessary condition for the reality of the spectrum of a
non-Hermitian Hamiltonian}, J. Math. Phys. {\bf 43} (2002), 205–-214.

\bibitem{Mo2} A. Mostafazadeh, {\it Pseudo-Hermiticity versus ${\cal PT}$--symmetry. {\rm II}. A complete characterization of non-Hermitian Hamiltonians with a real spectrum}, J. Math. Phys. {\bf 43} (2002), 2814–-2816.

\bibitem{Mo3} A. Mostafazadeh, {\it Pseudo-Hermiticity versus ${\cal PT}$--symmetry. {\rm III}. Equivalence of pseudo-Hermiticity and the presence of antilinear symmetries}, J. Math. Phys. {\bf 43} (2002), 3944–-3951.


\bibitem{Sj74} J. Sj\"{o}strand, {\it Parametrices for pseudodifferential operators with multiple characteristics},
Ark. f\"{o}r Mat. {\bf 12} (1974), 85--130.

\bibitem{Sj82} J. Sj\"ostrand, {\it Singularit\'es analytiques microlocales}, Ast\'erisque, {\bf 95} (1982), 1--166,
Soc. Math. France, Paris.

\bibitem{SjHokk} J. Sj\"ostrand, {\it Analytic wavefront sets and operators with multiple characteristics}, Hokkaido Math. Journal, {\bf 12} (1983), 392--433.

\bibitem{Sj86} J. Sj\"ostrand, {\it Semiclassical resonances generated by nondegenerate critical points},
Pseudodifferential operators (Oberwolfach, 1986), 402–-429, Lecture Notes in Math., {\bf 1256}, Springer, Berlin, 1987

\bibitem{Sj95} J. Sj\"ostrand, {\it Function spaces associated to global I-Lagrangian manifolds}, Structure of solutions of differential
equations, Katata/Kyoto, 1995, World Sci. Publ., River Edge, NJ (1996).

\bibitem{SAC} A. V. Sokolov, A. A. Andrianov, and F. Cannata: {\it Non-Hermitian quantum mechanics of non-diagonalizable Hamiltonians: puzzles with self-
orthogonal states}, J. Phys. A {\bf 39} (2006), 10207--10227

\bibitem{Viola2} J. Viola, {\it Non-elliptic quadratic forms and semiclassical estimates for non-selfadjoint operators}, preprint, 2011,
{\sf http://arxiv.org/abs/1109.5045}

\end{thebibliography}
\end{document}